\documentclass[twocolumn,showpacs,preprintnumbers,amsmath,amssymb]{revtex4}
\UseRawInputEncoding
\usepackage{mathrsfs}
\usepackage{amsfonts}
\usepackage{}
\usepackage{tikz}
\usepackage{graphicx}% Include figure files
\usepackage{dcolumn}% Align table columns on decimal point
\usepackage{bm}% bold math
\usepackage{amsmath,amsthm}
\usepackage{diagbox}
\usetikzlibrary{arrows.meta, positioning, calc}  % 
\usetikzlibrary{shapes.geometric}
\usetikzlibrary{decorations.pathreplacing, decorations.pathmorphing}
\usepackage{multirow}
\usepackage{booktabs}
\usepackage{xcolor} 
\usepackage{caption}    % 标题控制
\begin{document}

\preprint{}

\title
{Classical representation of local Clifford operators}% Force line breaks with \\

\author{Cai-Hong Wang$^{1}$}%
%\email[]{chwang@hpu.edu.cn}
\author{Jiang-Tao Yuan$^{1}$}%
\email[]{jtyuan@hpu.edu.cn}
\author{Zhi-Hao Ma$^{2}$}
\email[]{mazhihao@sjtu.edu.cn}
\author{Shao-Ming Fei$^{3}$}
\email[]{feishm@cnu.edu.cn}
\author{Shang-Quan Bu$^{4}$}
\email[]{bushangquan@tsinghua.edu.cn}

\affiliation{%
$^{1}$ Department of General Education, Wuxi University, Wuxi, 214105, China
$^{2}$School of Mathematical Sciences, MOE-LSC, Shanghai Jiao Tong University, Shanghai, 200240, China;\\
Shanghai Seres Information Technology Co., Ltd., Shanghai, 200040, China;
and Shenzhen Institute for Quantum Science and Engineering, Southern University of Science and Technology, Shenzhen, 518055, China\\
$^{3}$School of Mathematical Sciences, Capital Normal University, Beijing 100048, China\\
$^{4}$Department of Mathematical Sciences, University of Tsinghua, Beijing 100084, China
}%

\date{\today}% It is always \today, today,

\begin{abstract}
It is known that every (single-qudit) Clifford operator maps the full set of generalized Pauli matrices (GPMs) to itself under unitary conjugation, which is an important quantum operation and plays a crucial role in quantum computation and information. However, in many quantum information processing tasks,  it is required that a specific set of GPMs be  mapped to another such set under conjugation, instead of the entire set. We formalize this by introducing local Clifford operator, which maps a given $n$-GPM set to another such set under unitary conjugation. We establish necessary and sufficient conditions for such an operator to transform a pair of GPMs, showing that these local Clifford operators admit a classical matrix representation, analogous to the classical (or symplectic) representation of standard (single-qudit) Clifford operators. Furthermore, we demonstrate that any local Clifford operator acting on an $n$-GPM ($n\geq 2$) set can be decomposed into a product of standard Clifford operators and a local Clifford operator acting on a pair of GPMs. This decomposition provides a complete classical characterization of unitary conjugation mappings between $n$-GPM sets. As a key application, we use this framework to address the local unitary equivalence (LU-equivalence) of sets of generalized Bell states (GBSs). We prove that the 31 equivalence classes of $4$-GBS sets in bipartite system $\mathbb{C}^{6}\otimes \mathbb{C}^{6}$ previously identified via Clifford operators are indeed distinct under LU-equivalence, confirming that this classification is complete.
\end{abstract}

\maketitle

%\tableofcontents

\section{Introduction}
Quantum computation has emerged as a highly promising and dynamic field, marked by substantial theoretical and experimental advances in recent years~\cite{niel2010camb-book,kaye2007camb-book,goog2023nature0}. Clifford operators play a central role in many quantum information processing protocols, forming a fundamental class of quantum operations that underpin fault-tolerant quantum computation~\cite{gotts1997phd,gotts1998pra,goog2023nature}. According to the Gottesman-Knill theorem, the evolution of a quantum circuit composed solely of Clifford gates, when initialized in a computational basis state, can be efficiently simulated on a classical computer~\cite{aaro2004pra,berg2021iee}. Although Clifford gates alone are not universal for quantum computation, their universality can be achieved by incorporating a single non-Clifford gate, such as the $T$-gate. A prominent approach to realizing such universality is magic state distillation~\cite{seli2015lmcs,bravyi2005magic}. Moreover, Clifford operators are pivotal in various many quantum information protocols, including quantum error correction \cite{gras2003ifc} and entanglement distillation \cite{deha2003pra,zheng2025quantum}. Their unique mathematical properties establish Clifford operators as fundamental building blocks across quantum information science.

In this work, we focus on single-qudit Clifford operators\textemdash unitary operators that map the Pauli group to itself under conjugation. The classical representation of these operators have been extensively studied \cite{appl2005jmp,fari2014jpa,hos2005pra}. In a $d$-dimensional Hilbert space, it is known that, up to a global phase, the classical representation of a (single-qudit) Clifford operator corresponds to a two by two symplectic matrix over $\mathbb{Z}_d$.

In many quantum information processing tasks, it suffices to consider a subset of generalized Pauli matrices (GPMs), also known as Weyl-Heisenberg operators, rather than the full set of GPMs.  Notably, two sets of generalized Bell states (GBSs) that are local unitary equivalent (LU-equivalent) exhibit identical effectiveness and utility in quantum protocols \cite{wu-tian2018pra}. Specially, their distinguishability under local operations and classical communication (LOCC) is exactly the same. Due to the one-to-one correspondence between GBSs and GPMs, the LU-equivalence of GBS sets is generally determined by the unitary equivalence (U-equivalence) of the corresponding GPM sets. In \cite{wang2021jmp,wang2025epj}, the authors characterized the U-equivalence of GPM sets, showing that determining whether two GPM sets are U-equivalent reduces to determining whether they are unitary conjugate equivalent (UC-equivalent). Recall that two GPM sets are UC-equivalent if there exists a unitary operator that maps one set to the other under conjugation (up to a global phase). For convenience, we refer to such a unitary operator as a {\it local Clifford operator}. Thus, the problem of finding all GPM sets UC-equivalent to a given GPM set reduces to identifying all local Clifford operators acting on that set. Wu et al. \cite{wu-tian2018pra} constructed a class of unitary operators to prove LU-equivalence between two GBS sets. Although these operators are not standard Clifford operators, they satisfy the definition of local Clifford operators proposed in this work. Consequently, as an extension of Clifford operators, the study of local Clifford operators is both fundamentally important and practically necessary.

Inspired by the classical (or symplectic) representation of Clifford operators, in this work, we investigate the classical (or matrix) representation of local Clifford operators. We show that the classical representation of a local Clifford operator comprises a sequence of Clifford operators along with a local Clifford operator acting on a pair of GPMs. This representation, combined with the method in \cite{wang2021jmp,wang2025epj}, can identify all GPM sets U-equivalent to a given $n$-GPM set and fully characterize the U-equivalence classes of all $n$-GPM sets. Consequently, the local distinguishability problem for GBS sets reduces to determining the local distinguishability of representative elements from each equivalence class. This demonstrates the utility of local Clifford operators in addressing the local discrimination problem of GBS sets.

This paper is organized as follows. In Section II, We review fundamental concepts, including Clifford operators and their classical representation, while also introducing the notion of local Clifford operators. In Section III, for a $d$-dimensional quantum system, we present the classical representation of local Clifford operators acting on a binary GPM sets $\{X^a, Z^b\}$, where $a$ and $b$ are positive factors of the dimension $d$. In Section IV, we obtain the classical representation of local Clifford operators acting on a $n$-GPM set. It is shown that the classical representation of a local Clifford operator is composed of some Clifford operators, along with a local Clifford operator acting on a $2$-GPM set. In Section V, using classical representations, we provide procedures to determine the LU-equivalence class of a given GBS set and to verify whether two given GBS sets are LU-equivalent. Finally, we summarize our conclusions and propose future research.

\section{Preliminaries}

\newtheorem{definition}{Definition}
\newtheorem{lemma}{Lemma}
\newtheorem{theorem}{Theorem}
\newtheorem{corollary}{Corollary}
\newtheorem{example}{Example}
\newtheorem{proposition}{Proposition}
\newtheorem{problem}{Problem}
\newtheorem{conjecture}{Conjecture}
\newtheorem{remark}{Remark}

\def\QEDclosed{\mbox{\rule[0pt]{1.3ex}{1.3ex}}}
\def\QED{\QEDclosed}
\def\proof{\indent{\em Proof}.}
\def\endproof{\hspace*{\fill}~\QED\par\endtrivlist\unskip}

\newtheorem{procedure}{Procedure}

Consider a $d$-dimensional Hilbert space with the computational basis $\{|j\rangle\}_{j=0}^{d-1}$. Let $\mathbb{Z}_{d}=\{0,1,\ldots,d-1\}$ denote the integers modulo $d$. Let $X^{m}Z^{n}, m, n\in\mathbb{Z}_{d}$ be {\it generalized Pauli matrices} (GPMs) where $X|j\rangle=|j+1$ mod $d\rangle$, $Z|j\rangle=\omega^{j}|j\rangle$ and $\omega=e^{2\pi i/d}$, which are generalizations of Pauli matrices. For convenience, we denote a GPM $X^mZ^n$ by $(m, n)$.

In a bipartite quantum system $\mathbb{C}^{d}\otimes\mathbb{C}^{d}$, the canonical maximally entangled state $|\Phi\rangle$ is $|\Phi_{00}\rangle=(1/\sqrt{d})\sum_{j=0}^{d-1}|jj\rangle$. The states
\begin{eqnarray*}
|\Phi_{m,n}\rangle=(I\otimes X^{m}Z^{n})|\Phi\rangle
\end{eqnarray*}
are called {\it generalized Bell states} (GBSs). Note that there is a one-to-one correspondence between GBSs and GPMs as above.

A {\it Clifford operator} is defined as a unitary operator that maps all GPMs to GPMs (up to a global phase) under conjugation. Following from \cite{hos2005pra}, each single-qudit Clifford operator is, up to a global phase, uniquely represented by a two by two  symplectic matrix over $\mathbb{Z}_d$. In other words, a Clifford operator has a classical representation given by
\begin{eqnarray}\label{classical2.1}
&&W = \left[
\begin{array}{llll}
a_1 &b_1\\
a_2 &b_2
\end{array}
\right]
\end{eqnarray}
where the entries are over $\mathbb{Z}_d$ and $\det (W)=a_1b_2-a_2b_1\equiv 1$ (mod $d$). It performs the operation $X^sZ^t\mapsto X^{a_1s+b_1t}Z^{a_2s+b_2t}$, in other words, the action of the Clifford operator $W$ on a GPM $(s,t)$ is given by the linear transformation:
$$\left[\begin{array}{llll} a_1 &b_1\\ a_2 &b_2 \end{array} \right]\left[\begin{array}{llll} s\\ t \end{array} \right]=\left[\begin{array}{llll} a_1s+b_1t\\ a_2s+b_2t \end{array} \right].$$
In this paper, we use the same letter to denote a Clifford operator and its matrix representation.
Two LU-equivalent GBS sets demonstrate identical effectiveness and versatility. For many quantum information processing tasks, it suffices to consider the LU-equivalence between $n$-GBS sets (or the U-equivalence between $n$-GPM sets) rather than analyzing U-equivalence between all GPMs \cite{sing2017pra,tian2016pra,wu-tian2018pra,yuan2022quantum,wang2025epj,zhou2024pra}. This motivates the introduction of local Clifford operators.
\begin{definition}{\rm (Local Clifford operator)}
Let $\mathcal{M}=\{X^{s_1}Z^{t_1}, X^{s_2}Z^{t_2}, \cdots ,X^{s_n}Z^{t_n}  \}$ be a set of GPMs on $\mathbb{C}^d$. A local Clifford operator acting on $\mathcal{M}$
is a unitary operator on $\mathbb{C}^d$ that maps $\mathcal{M}$ to an $n$-GPM set under conjugation (up to a global phase).
\end{definition}
For two sets of unitary matrices $\mathcal{M}$ and $\mathcal{N}$, if there exists a unitary operator $U$ such that
\begin{eqnarray*}
U\mathcal{M}U^\dag\approx \mathcal{N},
\end{eqnarray*}
where "$\approx$" denotes equality up to a global phase, then they are called \textit{unitary conjugate equivalent} (UC-equivalent), denoted by $\mathcal{M}\mathop{\sim}\limits^U \mathcal{N}$. The transformation between $\mathcal{M}$ and $\mathcal{N}$ is called \textit{a unitary conjugate transformation} (UC-transformation). Obviously, if $\mathcal{M}$ and $\mathcal{N}$ are both sets of GPMs, then any local Clifford operator that maps $\mathcal{M}$ to $\mathcal{N}$ is precisely a UC-transformation between $\mathcal{M}$ and $\mathcal{N}$.

It is well known, for an element $x$ of a group, its order refers to the smallest positive integer $m$ with the property $x^{m}=I$ or the infinity if  $x^{m}\neq I$ for all $m\neq 0$. We denote the order of $x$ by $O(x)$. Since the global phase does not affect the LU-equivalence between quantum states,
in \cite{wang2025epj}, the authors introduced essential order and essential power, which are useful for studying LU-equivalence.
\begin{definition}[\cite{wang2025epj}]
Let $U$ be a GPM defined on $\mathbb{C}^d$. A positive integer $a$ is said to be the \textit{essential order} of $U$ if it is the smallest positive integer such that $U^a\approx I$, where "$\approx$" denotes equality up to a global phase. We denote the \textit{essential order} by $O_e(U)$. When $U\not\approx I$, we call $d/O_e(U)$ the \textit{essential power} of $U$ and denote it as $P_e(U)$. When $U\approx I$, let $P_e(U)=0$.
\end{definition}
For each GPM $X^sZ^t$ on $\mathbb{C}^d$, it is easy to know that its essential order
$$O_e(X^sZ^t)=\frac{d}{\gcd(s,t,d)}$$
and, in turn, its essential power
$$P_e(X^sZ^t)=\gcd(s,t,d),$$
where $\gcd(s,t,d)$ refers to the greatest common divisor of $s, t$ and $d$.

For example, let $X$ and $Z$ be Pauli matrices on $\mathbb{C}^6$, then $P_e(X^4)=\gcd(4,0,6)=2$, $P_e(X^2Z^2)=\gcd(2,2,6)=2$ and $P_e(X^3Z^4)=\gcd(3,4,6)=1$.

For a set of $n$-GPMs, it has a corresponding essential power vector \cite{wang2025epj}.
\begin{definition}[\cite{wang2025epj}]
Let $\mathcal{M}=\{M_1, M_2, \cdots , M_n \}$ be a GPM set and $x=( P_e(M_1), P_e(M_2), \cdots , P_e(M_n) )^T$ be a $n$-dimensional real vector.
The \textit{essential power vector} of $\mathcal{M}$ is defined as $P_e(\mathcal{M})=x^\uparrow$,
where the symbol $\uparrow$ means that the elements in the vector taking ascending order.
\end{definition}
For example, if $\mathcal{M}=\{X, X^6Z^6, X^9, X^4Z^2, X^5Z^8\}$ is a GPM set on $\mathbb{C}^{12}$,
then $P_e(X)=1$, $P_e(X^6Z^6)=6$, $P_e(X^9)=3$, $P_e(X^4Z^2)=2$, $P_e(X^5Z^8)=1$, and  $P_e(\mathcal{M})=(1, 1, 2, 3, 6)^T$.

In \cite{wang2025epj}, the authors showed a necessary and sufficient condition for the UC-equivalence between two GPMs.
\begin{lemma}[\cite{wang2025epj}]\label{lem2.1}
Any nontrivial GPM $X^sZ^t (\neq I)$ is UC equivalent to $Z^b$, where $b=\gcd(s,t,b)$ . Consequently, two GPMs are UC equivalent if and only if they have the same essential power.
\end{lemma}
See Figure \ref{fig2.1} for the visualization of Lemma \ref{lem2.1}.

\begin{figure}[htbp]
\centering
\begin{tikzpicture}[
    box/.style={rectangle, draw=blue, thick, minimum width=2cm, minimum height=1cm, fill=blue!10},
    input/.style={circle, draw=black, thick, fill=green!20, text width=1cm, minimum height=0.5cm, align=center},
    output/.style={circle, draw=black, thick, fill=green!20, text width=1cm, minimum height=0.5cm, align=center},
     cond/.style={ellipse, draw=red, thick, minimum width=1.5cm, minimum height=1.2cm, fill=red!10},
    arrow/.style={->, >=stealth, thick}
]

% Main operator
\node[box] (U) at (0,0) {\Large$UC$};

% Input and output
\node[input,left=1cm of U] (input) {$X^sZ^t$};
\node[output,right=1cm of U] (output) {$X^{s^{\prime}}Z^{t^{\prime}}$};

% Conditions
\node[cond] (cond1) at (0,1.5) {$\gcd(s,t,d) =\gcd(s^{\prime},t^{\prime},d)$};

% Arrows
\draw[arrow] (input) -- (U);
\draw[arrow] (U) -- (output);
\draw[arrow, red, dashed] (cond1) -- (U);

% Title
\end{tikzpicture}
\caption{There exists a local Clifford operator mapping the GPM $X^sZ^t$ to $X^{s^{\prime}}Z^{t^{\prime}}$ if and only if their essential powers are the same.}
\label{fig2.1}
\end{figure}

In Lemma \ref{lem2.1}, the UC-transformation from $X^sZ^t$ to $Z^b$ can be implemented by sequentially applying two Clifford operators.  Let $\gcd(s,t)=m$, $\gcd(m,d)=b$. There exist integers $p,q,p^{\prime},q^{\prime}$ such that $ps+qt=m, p^{\prime}d+q^{\prime}m=b$.
Let $C_{\gcd(s,t)}$ and $C_{\gcd(d,m)}$ be two Clifford operators with the corresponding classical representation
$$C_{\gcd(s,t)}=\left[ \begin{array}{llll} t/m &-s/m\\ \ p &\ \ q \end{array} \right],$$
$$C_{\gcd(d,m)}=\left[ \begin{array}{llll} m/b &-d/b\\ \ p^{\prime} &\ \ q^{\prime} \end{array} \right].$$
We get $X^sZ^t\mathop{\sim}\limits^{C_{\gcd(s,t)}} Z^m \mathop{\sim}\limits^{C_{\gcd(d,m)}} Z^b$, and denote the composite of $C_{\gcd(s,t)}$ and $C_{\gcd(d,m)}$ as $C_{\gcd(s,t,d)}$, which is a local Clifford operator that maps $X^sZ^t$ to $Z^{\gcd(s,t,d)}$.

According to Lemma \ref{lem2.1}, it is known that two UC-equivalent GPM sets have the same essential power vectors. This necessary condition for UC-equivalence can be used to determine that two GPM sets are not UC-equivalent.

In the next section, based on Lemma \ref{lem2.1}, we establish necessary and sufficient conditions for a 2-GPM set to be UC-equivalent to a special 2-GPM set $\{ X^{a}, Z^{b} \}$ (where $a$ and $b$ are positive factors of $d$) defined on $\mathbb{C}^{d}$, with $d\geq 2$ being an integer. These are the necessary and sufficient conditions for a local Clifford operator to act on a 2-GPM set $\{X^a, Z^b\}$.

\section{Local Clifford operators acting on 2-GPM sets}
In this section, we consider the local Clifford operators acting on a $2$-GPM set $\{X^a, Z^b\}$, where $a$ and $b$ are positive factors of $d$ with $a, b\neq d$.
The identity matrix is always mapped to itself under any local Clifford operator. Henceforth, unless otherwise stated, all GPM sets we consider exclude the identity matrix.

It is known that the essential power of a GPM and the commutation coefficient between two GPMs are invariant under UC-transformations.
Furthermore, it turns out that these two properties are essential features of local Clifford operators acting on $2$-GPM sets.

\begin{lemma}\label{lem3.1}
If $a$ and $b$ are both positive factors of $d$, then two nontrivial GPMs $X^a$ and $Z^b$ on $\mathbb{C}^d$ are UC-equivalent to $X^{ua}$ and $Z^b$ respectively through a unitary operator $W$ if and only if  $\gcd(u,d/a)=1$ and $uab\equiv ab$ (mod $d$).
\end{lemma}

Note that the condition $\gcd(u, d/a) = 1$ is equivalent to $\gcd(ua, d) = a$. This means the local Clifford operator $W$ preserves the essential power of $X^a$ (i.e., $X^a$ remains a GPM of the same essential power under conjugation by $W$). The commutation coefficient between two GPMs $X^a$ and $Z^b$ is $\omega^{ab}$, where $\omega$ is a primitive $d$-th root of unity. Similarly, the condition $uab \equiv ab \pmod{d}$ implies that the local Clifford operator $W$ preserves the commutation relation between $X^a$ and $Z^b$.
In other words, Lemma \ref{lem3.1} says that the preservation of both essential power and commutation coefficients completely characterizes such local Clifford operators acting on the 2-GPM sets $\{X^a, Z^b\}$.

When $ a = b = 1 $, the local Clifford operator in Lemma 2 is actually a Clifford operator, its characteristic condition simplifies to $ u \equiv 1 \pmod{d} $. This is precisely the feature of the matrix in the classical representation of a Clifford operator. Therefore, the characteristics of the local Clifford operator in Lemma 2 can be regarded as its classical representation by the matrix $\left[ \begin{array}{llll} u &0\\ 0 &1 \end{array} \right]$, where $uab\equiv ab\pmod{d}$.
See Figure \ref{fig3.1} for the schematic diagram of Lemma \ref{lem3.1}.

\begin{figure}[h]
\centering
\begin{tikzpicture}[
    box/.style={rectangle, draw=blue, thick, minimum width=2cm, minimum height=1cm, fill=blue!10},
    cond/.style={ellipse, draw=red, thick, minimum width=1.5cm, minimum height=1.2cm, fill=red!10},
    arrow/.style={->, >=stealth, thick}
]

% Main operator
\node[box] (U) at (0,0) {\Large$W$};

% Input and output
\node[left=1cm of U] (input) {$\{X^a, Z^b\}$};
\node[right=1cm of U] (output) {$\{X^{ua}, Z^{b}\}$};

% Conditions
\node[cond] (cond1) at (-2,1.5) {$\gcd(u, d/a) = 1$};
\node[cond] (cond2) at (2,1.5) {$uab \equiv ab \pmod{d}$};

% Arrows
\draw[arrow] (input) -- (U);
\draw[arrow] (U) -- (output);
\draw[arrow, red, dashed] (cond1) -- (U);
\draw[arrow, red, dashed] (cond2) -- (U);

% Title
\node[below=0.2cm of U ] {Visualization of Lemma \ref{lem3.1}: The classical representation};
\end{tikzpicture}
\caption{ There exists a local Clifford operator mapping the GPM pair $\{X^a, Z^b\}$ to $\{X^{ua}, Z^b\}$ if and only if the two conditions in red ellipses hold.}
\label{fig3.1}
\end{figure}

\begin{proof}
First, we demonstrate that the two conditions, preservation of essential power and commutation relations, are sufficient to construct a local Clifford operator.  If $\gcd(u,d/a)=1$ and $uab\equiv ab$ (mod $d$), we define an operator $W$ as follows
\begin{eqnarray*}
W|i+ja\rangle = |i+uja\rangle, 0\leq i < a, 0\leq j < d/a.
\end{eqnarray*}
If $i+uja=i'+uj'a$ (mod $d$), that is, $i-i'\equiv u(j'-j)a$ (mod $d$), then $i=i', j=j'$ because of $\gcd(u,d/a)=1$ and {\color{blue}$-a< i-i'< a$}. Thus $W$ is a permutation operator and is unitary.
According to the condition $uab\equiv ab$ (mod $d$), we have
\begin{align*}
WZ^bW^\dag=&\sum_{i,j,i',j',i'',j''}\omega ^{b(i'+j'a)} |i+uja\rangle \langle i+ja|i'+j'a\rangle \\
&\langle i'+j'a|i''+j''a\rangle \langle i''+uj''a| \\
=&\sum_{i,j}\omega ^{b(i+ja)} |i+uja\rangle \langle i+uja|\\
=&\sum_{i,j}\omega ^{b(i+uja)} |i+uja\rangle \langle i+uja|=Z^b,\\
WX^aW^\dag=&\sum_{i,j,i',j',i'',j''} |i+uja\rangle \langle i+ja|i'+(j'+1)a\rangle \\
&\langle i'+j'a|i''+j''a\rangle \langle i''+uj''a| \\
=&\sum_{i,j} |i+u(j+1)a\rangle \langle i+uja|=X^{ua}.
\end{align*}
So the two GPMs $X^a$ and $Z^b$ are UC equivalent to $X^{ua}$ and $Z^b$ respectively through the unitary operator $W$.

Conversely, for an arbitrary local  Clifford operator, it can be directly verified that it preserves both the essential power and commutation relations. The detailed proof is provided in Appendix A.
\end{proof}

Note that the unitary operator $W$ in Lemma \ref{lem3.1} is a local Clifford operator, but not necessarily a Clifford operator.
For example, for $a=6$, $b=12$, $u=5$, and $d=72$, it is clear that
\[
\gcd(5, 72/6) = 1 \quad \text{and} \quad 5 \times 6 \times 12 \equiv 6 \times 12 \pmod{72}.
\]
Then, Lemma~\ref{lem3.1} implies that the GPMs $X^6$ and $Z^{12}$ on $\mathbb{C}^{72}$ are UC-equivalent to $X^{30}$ and $Z^{12}$, respectively, via a unitary operator $W$. However, $W$ is not a Clifford operator because the matrix in the classical representation does not satisfy the condition $ u \equiv 1 \pmod{d} $.

By Eq. \eqref{eq6} in Appendix A, if two general GPMs are UC-equivalent to $X^{a}$ and $Z^{b}$ respectively, they necessarily take the forms $X^{\mu a}Z^{\sigma a}$ and $X^{\eta b}Z^{\nu b}$. Based on Lemma \ref{lem3.1}, we can now present the characteristics (or classical representation) of a general local Clifford operator acting on two GPMs $X^{a}$ and $Z^{b}$.

\begin{theorem}\label{th3.1}
If $a$ and $b$ are positive factors of $d$,
then two nontrivial GPMs $X^{a}$ and $Z^{b}$ are UC equivalent to
$X^{u_{1} a}Z^{v_{1} a}$ and $X^{u_{2} b}Z^{v_{2} b}$ respectively
if and only if
\begin{eqnarray}\label{eq1th3.1}
\gcd(u_{1}, v_{1}, d/a)=\gcd(u_{2}, v_{2}, d/b)=1,
\end{eqnarray}
\begin{eqnarray}\label{eq2th3.1}
(u_{1}v_{2}-u_{2}v_{1})ab\equiv ab \pmod{d},
\end{eqnarray}
\begin{eqnarray}\label{eq3th3.1}
\gcd(u_{1}v_{2}-u_{2}v_{1}, d/a, d/b)=1,
\end{eqnarray}
where $0 \leq u_{1}, v_{1}< d/a$ and  $0\leq u_{2}, v_{2}< d/b$.
\end{theorem}
See Figure \ref{fig3.2} for the visualization of Theorem \ref{th3.1}.

\begin{figure}[htbp]
\centering
\begin{tikzpicture}[
    box/.style={rectangle, draw=blue, thick, minimum width=2cm, minimum height=1cm, fill=blue!10},
    cond/.style={ellipse, draw=red, thick, minimum width=1.5cm, minimum height=1.2cm, fill=red!10},
    arrow/.style={->, >=stealth, thick}
]

% Main operator
\node[box] (U) at (-0.5,0) {\large$UC$};

% Input and output
\node[left=0.5cm of U] (input) {$\{X^a, Z^b\}$};
\node[right=0.5cm of U] (output) {$\{X^{u_1a}Z^{v_1a}, X^{u_2b}Z^{v_2b}\}$};

% Conditions
\node[cond] (cond1) at (-0.5,2.8) {$\gcd(u_1,v_1,d/a)=1=\gcd(u_{2}, v_{2}, d/b)$};
\node[cond] (cond2) at (2,1.5) {$(u_1v_2-u_2v_1)ab \equiv ab$};
\node[cond] (cond3) at (-0.5,-1.5) {$\gcd(u_1v_2-u_2v_1,d/a,d/b)=1$};

% Arrows
\draw[arrow] (input) -- (U);
\draw[arrow] (U) -- (output);
\draw[arrow, red, dashed] (cond1) -- (U);
\draw[arrow, red, dashed] (cond2) -- (U);
\draw[arrow, red, dashed] (cond3) -- (U);

% Title
\node[below=0.2cm of cond3] {Visualization of Theorem \ref{th3.1}: The classical representation};
\end{tikzpicture}

\caption{ There exists a local Clifford operator mapping the GPM pair $\{X^a, Z^b\}$ to $\{X^{u_{1} a}Z^{v_{1} a}, X^{u_{2} b}Z^{v_{2} b}\}$ if and only if the three conditions in red ellipses hold.}
\label{fig3.2}
\end{figure}

When $ a = b = 1 $, the characteristic conditions \eqref{eq1th3.1}-\eqref{eq3th3.1} simplify to $u_{1}v_{2}-u_{2}v_{1}\equiv 1 \pmod{d} $. Thus, the local Clifford operator in Theorem \ref{th3.1} is simply a Clifford operator with the classical representation
$\left[ \begin{array}{llll} u_{1} &u_{2}\\ v_{1} &v_{2} \end{array} \right]$, where $u_{1}v_{2}-u_{2}v_{1}\equiv 1\pmod{d}$.
See Figure \ref{fig3.3} for the schematic diagram of a Clifford operator. 

\begin{figure}[htbp]
\centering
\begin{tikzpicture}[
    box/.style={rectangle, draw=blue, thick, minimum width=2cm, minimum height=1cm, fill=blue!10},
    input/.style={circle, draw=black, thick, fill=green!20, text width=1cm, minimum height=0.5cm, align=center},
    output/.style={circle, draw=black, thick, fill=green!20, text width=1cm, minimum height=0.5cm, align=center},
     cond/.style={ellipse, draw=red, thick, minimum width=1.5cm, minimum height=1.2cm, fill=red!10},
    arrow/.style={->, >=stealth, thick}
]

% Main operator
\node[box] (U) at (-1.5,0) {\Large$UC$};

% Input and output
\node[left=0.5cm of U] (input) {$\{X, Z\}$};
\node[right=0.5cm of U] (output) {$\{X^{u_1}Z^{v_1}, X^{u_2}Z^{v_2}\}$};

% Conditions
\node[cond] (cond1) at (0,1.5) {$u_{1}v_{2}-u_{2}v_{1}\equiv 1 \pmod{d} $};

% Arrows
\draw[arrow] (input) -- (U);
\draw[arrow] (U) -- (output);
\draw[arrow, red, dashed] (cond1) -- (U);

% Title
\node[below=0.2cm of U](title) at (-0.2,-0.5) {Schematic diagram of a Clifford operator};

\end{tikzpicture}
\caption{There exists a Clifford operator mapping the GPM pair $\{X, Z\}$ to $\{X^{u_{1}}Z^{v_{1}}, X^{u_{2}}Z^{v_{2}}\}$ if and only if $u_{1}v_{2}-u_{2}v_{1}\equiv 1 \pmod{d} $.}
\label{fig3.3}
\end{figure}

Therefore, the characteristics of the local Clifford operator in Theorem \ref{th3.1} can be described by its classical representation matrix
$\left[ \begin{array}{llll} u_{1} &u_{2}\\ v_{1} &v_{2} \end{array} \right]$, where $(u_{1}v_{2}-u_{2}v_{1})ab\equiv ab\pmod{d}$. For convenience, we denote the local Clifford operator in Theorem \ref{th3.1} by $L_{(a,b)}$. The conditions $0\leq u_{1},v_{1}< d/a$ and $0\leq u_{2},v_{2}< d/b$ imply that the number of local Clifford operators  on $\{ X^{a}, Z^{b} \}$ is finite.

Here is an outline of the proof of Theorem \ref{th3.1}, with the detailed proof provided in Appendix B:
First, the UC-equivalence between GPM pair $\{X^{u_{1} a}Z^{v_{1} a},\ X^{u_{2} b}Z^{v_{2} b}\}$ and $\{X^{a},\ Z^{b}\}$ can be derived from the UC-equivalence between $\{X^{u_{1} a}Z^{v_{1} a},\ X^{u_{2} b}Z^{v_{2} b}\}$ and $\{X^{u a}Z^{v a},\ Z^{b}\}$, as well as the UC-equivalence between $\{X^{a},\ Z^{b}\}$ and $\{X^{u a}Z^{v a},\ Z^{b}\}$ with $u\equiv u_{1}v_{2}-u_{2}v_{1}$ (mod $d/b$).
Second, by examining the conditions required for the UC-equivalence between $\{X^{u_{1} a}Z^{v_{1} a},\ X^{u_{2} b}Z^{v_{2} b}\}$ and $\{X^{u a}Z^{v a},\ Z^{b}\}$ and those for the UC-equivalence between $\{X^{a},\ Z^{b}\}$ and $\{X^{u a}Z^{v a},\ Z^{b}\}$, we combine these two sets of conditions to form the necessary and sufficient conditions for the UC-equivalence between $\{X^{u_{1} a}Z^{v_{1} a},\ X^{u_{2} b}Z^{v_{2} b}\}$ and $\{X^{a},\ Z^{b}\}$.
As a result, the necessary and sufficient conditions for the UC-equivalence between $\{X^{u_{1} a}Z^{v_{1} a},\ X^{u_{2} b}Z^{v_{2} b}\}$ and $\{X^{a},\ Z^{b}\}$ consist of three conditions, that is, one more than the simpler necessary and sufficient conditions given in Lemma 2.

It should be noted that the condition, Eq. (\ref{eq3th3.1}), is indispensable in Theorem \ref{th3.1}, as it cannot be derived from Eqs. (\ref{eq1th3.1}) and (\ref{eq2th3.1}). For example, when $d=72, a=6, b=12$, take $u_1=2, v_1=11, u_2=4, v_2=5$. Obviously, they satisfy Eqs. (\ref{eq1th3.1}) and (\ref{eq2th3.1}), but $u_{1}v_{2}-u_{2}v_{1}=-34$ and $\gcd(u_{1}v_{2}-u_{2}v_{1}, d/a, d/b)=\gcd(-34, 12, 6)=2$. Then Eq. (\ref{eq3th3.1}) does not hold.

Next, we consider two simple special cases.
\begin{remark}
For a GPM set $\mathcal{M}_z=\{Z^{t_1}, Z^{t_2}, \cdots ,Z^{t_n}  \}$, let $b=\gcd(t_1, t_2, \cdots , t_n, d)$, then it is clear that a unitary operator is a local Clifford operator acting on $\mathcal{M}_z$ if and only if it maps $Z^b$ to a GPM. By Lemma \ref{lem2.1}, we know that $Z^b$ is UC-equivalent to the GPM $X^{ub}Z^{vb}$ with $\gcd(u,v,d/b)=1$ through the Clifford operator $C^{-1}_{\gcd(ub,vb,d)}$. Then, all local Clifford operators on $\mathcal{M}_z$ can be realized through Clifford operators.

For a GPM set $\mathcal{M}_x=\{X^{s_1}, X^{s_2}, \cdots, X^{s_n}  \}$, it is known that $X$ and $Z$ can be interchanged via the Clifford operators 
$\left[ \begin{array}{ccc} 0 &1\\ -1 &0 \end{array} \right]$ or $\left[ \begin{array}{ccc} 0 &-1\\ 1 &0 \end{array} \right]$. 
Therefore, all local Clifford operators on $\mathcal{M}_x$ can also be realized using Clifford operators. In subsequent analyses, these two cases will no longer be considered.
\end{remark}

In the next section, we present the classical representation of local Clifford operators acting on an arbitrary $n$-GPM set with $n\ge 2$.

\section{The classical representation of local Clifford operators}

In this section, we investigate local Clifford operators acting on any $n$-GPM set defined on $C^{d}$. It turns out that an arbitrary local Clifford operator on a given $n$-GPM set can be decomposed into a sequence of Clifford operators, followed by a local Clifford operator acting on a special 2-GPM set (as in Theorem \ref{th3.1}). The classical representation of a general local Clifford operator is given in \eqref{classical4.1} below.

\subsection{For any dimension $d$}
In this subsection, we consider quantum systems of dimension $d$. Let $\mathcal{M}=\{X^{s_1}Z^{t_1}, X^{s_2}Z^{t_2}, \cdots ,X^{s_n}Z^{t_n}  \}$ be a GPM set with essential power vector $P_e(\mathcal{M})=(d_1,d_2,\cdots ,d_n)^T$. Without loss of generality, assume $d_1$ is the essential power of $X^{s_1}Z^{t_1}$. Applying the Clifford operator $C_{\gcd(s_1,t_1,d)}$, we can map $\mathcal{M}$  to
$\mathcal{M}^{\prime}=\{Z^{d_1}, X^{s_2^{\prime}}Z^{t_2^{\prime}}, \cdots ,X^{s_n^{\prime}}Z^{t_n^{\prime}}  \}$.
If $d_1$ satisfies the following condition
\begin{eqnarray*}
\text{Condition (I)}: d_1\mid t_2^{\prime}, \cdots ,t_n^{\prime},
\end{eqnarray*}
then by taking $a=\gcd(s_2^{\prime}, \cdots ,s_n^{\prime}, d)$ and $b=d_1$, the GPM sets $\mathcal{M}^{\prime}$ and $\{X^a,Z^b \}$ share identical local Clifford operators (see Lemma \ref{lemappd} in Appendix C for details). In other words, a unitary operator is a local Clifford operator on $\mathcal{M}^{\prime}$ if and only if it is a local Clifford operator on $\{X^a,Z^b \}$.  Therefore, every local Clifford operator on $\mathcal{M}$ has the form
$$L_{(a,b)}\circ C_{\gcd(s_1,t_1,d)},$$
where $L_{(a,b)}$ denotes a local Clifford operator on $\{X^a,Z^b \}$ as in Theorem \ref{th3.1}. If the GPM set $\mathcal{M}'$  does not satisfy Condition (I), it can always be transformed\textemdash using the {\it division algorithm}\textemdash into a GPM set that satisfies Condition (I). In fact, let $t_2^{\prime}=p_2d_1+r_2,\cdots ,t_n^{\prime}=p_nd_1+r_n$, where $0\leq r_2, \cdots ,r_n < d_1$ and $r_2, \cdots ,r_n$ are not all zero, and $\mathcal{M}_1=\{ Z^{d_1}, X^{s_2}Z^{r_2}, \cdots ,X^{s_n}Z^{r_n}  \}$. It is easy to see GPM sets $\mathcal{M}_1$ and $\mathcal{M}'$ share identical local Clifford operators. Consider $\mathcal{M}_1$ as $\mathcal{M}$ at the beginning of this section and repeat the above process until the resulting GPM set satisfies Condition (I). 
We can summarize the above discussion on any $n$-GPM set $\mathcal{M}$ with a schematic diagram (see Figure \ref{fig4.1}).

\begin{figure}[htbp]
\centering
\begin{tikzpicture}[
    node distance=0.8cm,
    op/.style={draw, rectangle, minimum width=1.2cm, minimum height=0.8cm, fill=blue!20},
    set/.style={ellipse, draw, minimum width=2cm, minimum height=0.8cm, fill=green!20},
    arrow/.style={-Stealth, thick}
]

% 输入集合 M
\node[set] (1) {Let $\mathcal{M}$ be a $n$-GPM set};

% 中间集合M'
\node[set, below=of 1] (2) {$\mathcal{M}^{\prime}$ does not satisfy Condition (I)};
\draw[arrow] (1) -- node[right] {applying a Clifford operator $C_{0}$} (2);

% 中间集合M_1
\node[set, below=of 2] (3) {$\mathcal{M}_{1}$};
\draw[arrow] (2) -- node[right] {division algorithm} (3);

% 中间集合
\node[set, below=of 3] (4) {$\mathcal{M}^{\prime}_{1}$ does not satisfy Condition (I)};
\draw[arrow] (3) -- node[right] { $C_{1}$} (4);

% repeating
\node[op, below=of 4, text width=6cm, align=center] (5) {Apply the division algorithm and Clifford operators repeatedly and in sequence};
\draw[arrow] (4) -- node[right] { } (5);

% 集合满足条件
\node[set, below=of 5] (6) {$\mathcal{M}^{\prime}_{m}$ satisfying Condition (I)};
\draw[arrow] (5) -- node[right] { } (6);

\end{tikzpicture}
\caption{Schematic diagram of the procedure for transforming a given $n$-GPM set, $\mathcal{M}$, into a GPM set, $\mathcal{M}^{\prime}_{m}$, satisfying Condition (I).}
\label{fig4.1}
\end{figure}

The local Clifford operator on the $n$-GPM set $\mathcal{M}^{\prime}_{m}$ is  a local Clifford operator  $L_{(a,b)}$ as in Theorem \ref{th3.1}. Thus, each local Clifford operator on the GPM set $\mathcal{M}$ has the form
\begin{eqnarray}\label{classical4.1}
L_{(a,b)}\circ C_m\circ\cdots \circ C_0,
\end{eqnarray}
where $C_{i}$ ($i=0,\cdots,m$) is a Clifford operator on $\mathcal{M}_{i}$ (as in Lemma \ref{lem2.1}). Consequently, a local Clifford operator is classically represented as a product of a series of two by two symplectic matrices and a specific two by two matrix determined by Theorem \ref{th3.1}. Since the number of possible operators $L_{(a,b)}$ is finite, we can use the classical representations (\ref{classical2.1}) and (\ref{classical4.1}) to:
(i) Generate all $n$-GPM sets that are UC-equivalent to a given $n$-GPM set, and
(ii) Determine whether any two given $n$-GPM sets are UC-equivalent to each other. 

We now present an example demonstrating how to derive the classical representation (\ref{classical4.1}) using the schematic diagram shown in Figure \ref{fig4.1}.
\begin{example}
For the quantum system $\mathbb{C}^{12}$, let $\mathcal{M}= \{ X^2Z^2, Z^3, X^4Z^8, X^6 \}$ be a $4$-GPM set, then $P_e(\mathcal{M})=\{2, 3, 4, 6 \}$. Applying the Clifford operator
$$C_{gcd(2,2)}=\left[ \begin{array}{llll} 1 &-1\\ 2 &-1 \end{array} \right]$$
to $\mathcal{M}$, we obtain the GPM set $\mathcal{M}^{\prime}= \{ Z^2, X^9Z^9, X^8, X^6 \}$. Clearly, the GPM set $\mathcal{M}^{\prime}$ does not satisfy Condition (I).
Using the division algorithm, we obtain a GPM set $\mathcal{M}_1= \{ Z^2, X^9Z^1, X^8, X^6 \}= \{  X^9Z^1, Z^2, X^8, X^6 \}$. Applying the Clifford operator
$$C_{\gcd(9,1)}=\left[ \begin{array}{llll} 1 &3\\ 1 &4 \end{array} \right]$$
to $\mathcal{M}_1$, we obtain the GPM set
$$\mathcal{M}_1^{\prime}= \{ Z, X^6Z^8, X^8Z^8, X^6Z^6 \}.$$
The GPM set $\mathcal{M}_1^{\prime}$ clearly satisfies Condition (I), then, by taking $a=\gcd(6,8,6,12)=2$ and $b=1$, any local Clifford operator acting on $\mathcal{M}$ can be expressed in the form
$$ L_{(2,1)}\circ C_{\gcd(9,1)}\circ C_{\gcd(2,2)}.$$
According to Theorem \ref{th3.1}, the operator $L_{(2,1)}$ has the matrix representation $\left[ \begin{array}{llll} u_{1} &u_{2}\\ v_{1} &v_{2} \end{array} \right],$ where $\gcd(u_1, v_1, 6)=\gcd(u_2, v_2, 12)=1$, $u_1v_2-u_2v_1\equiv 1 \pmod{6}$, and $0\leq u_1, v_1 < 6$, $0\leq u_2, v_2 < 12$. 
\end{example}

\subsection{$d=p^\alpha $}
In this subsection, we consider the case where $d=p^\alpha$, with $p$ prime and $\alpha$ a positive integer.
It turns out that any local Clifford operator can be represented as a composition of a Clifford operator and a local Clifford operator on a 2-GPM set, see (\ref{classical4.2}).

Consider a GPM set $\mathcal{M}= \{ X^{s_1}Z^{t_1}, \dots , X^{s_n}Z^{t_n} \}$, with essential power vector $P_e(\mathcal{M})=( d_1, \dots , d_n )$, where each $d_i$ is a prime power. Suppose that $d_1=p^\beta$ (where $0\leq \beta <\alpha$). Without loss of generality, we may assume $P_e(X^{s_1}Z^{t_1})=p^\beta$. Then, the Clifford operator $C_{\gcd(s_1, t_1, p^\alpha)}$ transforms $\mathcal{M}$ into the GPM set $\mathcal{M}'= \{ Z^{p^\beta},  X^{s'_2}Z^{t'_2}, \cdots , X^{s'_n}Z^{t'_n} \}$. Since UC-transformations preserve essential powers,
$$ \gcd(s'_i, t'_i, p^\alpha)=P_e(X^{s'_i}Z^{t'_i})=d_i\geq p^\beta .$$
Hence $p^\beta |t'_i$ ($i=2, \cdots , n$) and the GPM set $\mathcal{M}'$ satisfies Condition (I).
Then, taking  $p^\gamma=\gcd(s'_2, \cdots, s'_n, p^\alpha)$ with $0\leq \gamma <\alpha$, any local Clifford operator on $\mathcal{M}$ can be expressed in the form
\begin{eqnarray}\label{classical4.2}
L_{(p^\gamma,p^\beta)}\circ C_{\gcd(s_1, t_1, p^\alpha)},
\end{eqnarray}
where $L_{(p^\gamma,p^\beta)}$, as in Theorem \ref{th3.1}, is a local Clifford operator acting on $\{ X^{p^\gamma}, Z^{p^\beta} \}$. Consequently, every such operator admits a classical representation as a product of at most a two by two symplectic matrix and a specific  two by two matrix.

Since an integer is coprime with $p^t$ ($t > 0$) if and only if it is coprime with $p$, the conditions on local Clifford operators in Theorem \ref{th3.1} can be simplified. That is, the GPMs $X^{p^\gamma}$ and $Z^{p^\beta}$ are UC-equivalent to GPMs $X^{u_1 p^{\gamma}}Z^{v_1 p^{\gamma}}$ and $X^{u_2 p^\beta}Z^{v_2 p^\beta}$, respectively, if and only if the variables $u_i, v_i(i=1,2)$ satisfy the following conditions: $0 \leq u_1, v_1 < p^{\alpha -\gamma}$ and  $0\leq u_2, v_2  < p^{\alpha -\beta}$, and 
\begin{enumerate}
\item[{\rm(a)}] If $\beta +\gamma < \alpha$, then $(u_1v_2 -u_2v_1)p^{\beta +\gamma} \equiv p^{\beta +\gamma}\pmod{p^{\alpha}}.$
\item[{\rm(b)}] If $\beta +\gamma \geq \alpha$, then $\gcd(u_1v_2 -u_2v_1 , p)=1$.
\end{enumerate}
In particular, when $d=p^{\alpha}$ and $\beta=\gamma =0$, the local Clifford operator in Theorem \ref{th3.1} reduces to a Clifford operator with the classical representation
$\left[ \begin{array}{llll} u_{1} &u_{2}\\ v_{1} &v_{2} \end{array} \right]$ satisfying the condition $u_{1}v_{2}-u_{2}v_{1}\equiv 1\pmod{p^{\alpha}}$.

We now present an example demonstrating how to derive the classical representation (\ref{classical4.2}) using the schematic diagram shown in Figure \ref{fig4.1}.

\begin{example}
For the dimension $d=3^4$, let $\mathcal{M}= \{ X^2Z^4, X^3, X^9Z^{18}, Z^6 \}$ with $P_e(\mathcal{M})=\{ 1, 3, 3, 9 \}$. Applying the Clifford operator
$$C_{\gcd(2,4,3^4)}=\left[ \begin{array}{ccc} 2 &\ 0\\ -1 &41 \end{array} \right]\left[ \begin{array}{ccc} \ \ 2 &-1\\ -1 &\ \ 1 \end{array} \right]=\left[ \begin{array}{ccc} \ 4 &-2\\ 38 &\ 42 \end{array} \right]$$
to $\mathcal{M}$, we obtain the GPM set
$$\mathcal{M}'=\{ Z, X^{12}Z^{33}, Z^{25}, X^{69}Z^{9}  \}.$$
The GPM set $\mathcal{M}'$ clearly satisfies Condition (I).
Since $\gcd(12, 0, 69, 3^4)=3$ and $\beta +\gamma=0+1<4= \alpha$, any local Clifford operator on $\mathcal{M}$  can be expressed in the form
$$L_{(3, 1)}\circ C_{\gcd(2,4,3^4)}$$
with $L_{(3, 1)}=\left[ \begin{array}{llll} u_1 &u_2\\ v_1 &v_2 \end{array} \right]$, where $u_1v_2 -u_2v_1 \equiv 1\pmod{3^3}$ and $0 \leq u_1, v_1 < 3^3$, $0\leq u_2, v_2 < 3^4$.
\end{example}
It is evident that for  a quantum system with prime dimension $p$ (i.e., for the Hilbert space  $\mathbb{C}^p$), the local Clifford operators coincide with the Clifford operators.

\subsection{$d=pq$}
In this subsection, we consider the case where $d=pq$, with $p$ and $q$ being distinct primes and $p< q$. We show that any local Clifford operator can be represented as a product of at most two Clifford operators and a local Clifford operator acting on a specific binary GPM set, see (\ref{classical4.3})-(\ref{classical4.6}).

Consider a GPM set $\mathcal{M}= \{ X^{s_1}Z^{t_1}\cdots , X^{s_n}Z^{t_n} \}$ with $P_e(\mathcal{M})=( d_1, \dots , d_n )$. Suppose that $P_e(X^{s_1}Z^{t_1})=d_1$. If $d_1=1$, then the local Clifford operators on $\mathcal{M}$ have the form
\begin{eqnarray}\label{classical4.3}
L_{(a,1)}\circ C_{\gcd(s_1,t_1,pq)}.
\end{eqnarray}

If $d_1=p$, then there are two cases. One is $P_e(M)=(p, p, \dots , p)$. The local Clifford operators have the form
\begin{eqnarray}\label{classical4.4}
L_{(p,p)}\circ C_{\gcd(s_1,t_1,pq)}.
\end{eqnarray}
The other case is $P_e(M)=(p, p, \dots p, q, \dots , q)$. The local Clifford operators have the form (see Appendix D for details)
\begin{eqnarray}\label{classical4.5.1}
L_{(a,p)}\circ C_{\gcd(s_1,t_1,pq)}\ {\rm or} \\
L_{(a,1)}\circ C_{\gcd(s'_n, r_n, pq)}\circ C_{\gcd(s_1,t_1,pq)}.
\label{classical4.5.2}
\end{eqnarray}

If $d_1=q$, then $P_e(\mathcal{M})=(q, q, \cdots , q)$ and the local Clifford operators have the form
\begin{eqnarray}\label{classical4.6}
L_{(q,q)}\circ C_{\gcd(s_1,t_1,pq)}.
\end{eqnarray}

We thus conclude that in a $pq$-dimensional quantum system,  any local Clifford operator can be represented as a product of at most two Clifford operators and a local Clifford operator on a specific binary GPM set. That is,  every such operator admits a classical representation as a product of  at most two $2\times 2$ symplectic matrices and a specific $2\times 2$ matrix.

We now present an example demonstrating how to derive the classical representation using the schematic diagram shown in Figure \ref{fig4.1}.

\begin{example}
For $d=3\times 5$, let $\mathcal{M}=\{ X^6Z^6, Z^5, X^5Z^5  \}$, then $P_e(\mathcal{M})=(3, 5, 5)$. Applying the Clifford operator
$$C_{\gcd(6,6,15)}=\left[ \begin{array}{llll} 2 &-5\\ 1 &-2 \end{array} \right]\left[ \begin{array}{llll} 1 &-1\\ 1 &\ \ 0 \end{array} \right]=\left[ \begin{array}{llll} -3 &-2\\ -1 &-1 \end{array} \right]$$
to $\mathcal{M}$, we get $\mathcal{M}'=\{ Z^3, X^5Z^{10}, X^5Z^5  \}$. By using the division algorithm, we can obtain $\mathcal{M}_1=\{ Z^3, X^5Z, X^5Z^2  \}$ with $P_e(\mathcal{M}_1)=(1, 1, 3)$. Applying $C_{\gcd(5,2)}=\left[ \begin{array}{llll} 2 &-5\\ 1 &-2 \end{array} \right]$ to $\mathcal{M}_1$, we get a GPM set
$$\mathcal{M}'_1=\{ Z, X^5Z^3, Z^9  \}.$$
Evidently, $\mathcal{M}'_1$ satisfies Condition (I) and $a=5, b=1$. Hence the local Clifford operators on $\mathcal{M}$ have the form $L_{(5,1)}\circ C_{\gcd(5,2)}\circ C_{\gcd(6,6,15)}$ with $L_{(5, 1)}=\left[ \begin{array}{llll} u_1 &u_2\\ v_1 &v_2 \end{array} \right]$, where $\gcd(u_{2}, v_{2}, 15)=1$, $u_1v_2 -u_2v_1 \equiv 1\pmod{3}$ and $0 \leq u_1, v_1 < 3$, $0\leq u_2, v_2 < 15$.

\end{example}

\section{LU-equivalence of GBS sets}
In this section, we present an application of local Clifford operators: determining whether any two given GBS sets are LU-equivalent or whether their corresponding GPM sets are U-equivalent (defined below). An example is provided to illustrate how to determine whether two given GPM sets are U-equivalent.

Two GPM sets $\mathcal{M}$ and $\mathcal{N}$ are said to be {\it U-equivalent} if there exist unitary operators $U$ and $V$ such that 
$$\mathcal{M}\approx U\mathcal{N} V.$$ 
Since every GPM set is U-equivalent to a GPM set containing the identity matrix (by left-multiplying with the conjugate transpose of one of its elements), in this section, we restrict our consideration to sets that include the identity matrix. Such a GPM set is referred to as a {\it standard} GPM set.

The following lemma from \cite{wang2025epj} provides a useful characterization of U-equivalence between two $n$-GPM sets.
\begin{lemma}[\cite{wang2025epj}]\label{lem5.1}
For a given n-GPM set $\mathcal{M} =\{ M_1, M_2, \dots , M_n \}$, the standard GPM set $\mathcal{N}$ that is U-equivalent to the set $\mathcal{M}$ has the form
$$\mathcal{N} \approx \{ R(M_{i_0}^{\dag } M_i)R^{\dag } \}_{i=1}^n,$$
where $R$ is a unitary operator.
\end{lemma}

See Figure \ref{fig5.1} for the visualization of Lemma \ref{lem5.1}.

\begin{figure}[htbp]
\centering
\begin{tikzpicture}[
    node distance=0.8cm,
    op/.style={draw, rectangle, minimum width=1.2cm, minimum height=0.8cm, fill=blue!10},
    set/.style={ellipse, draw, minimum width=2cm, minimum height=1.2cm, fill=green!20},
    arrow/.style={-Stealth, thick}
]

% 输入集合 M
\node[set] (M) {$\mathcal{M} = \{M_1, M_2, \dots, M_n\}$};

% 中间集合 M_{i0}^†M
\node[set, below=of M] (MdaggerM) {$M_{i_0}^{\dagger}\mathcal{M}$};
\draw[arrow] (M) -- node[right] {left-multiplied by $M_{i_0}^{\dagger}$} (MdaggerM);

% 酉算子 R
\node[set, below=of MdaggerM] (N) {$\mathcal{N} \approx \{R(M_{i_0}^{\dagger}M_i)R^{\dagger}\}$};
\draw[arrow] (MdaggerM) -- node[right] {unitary conjugation by $R$} (N);

\end{tikzpicture}
\caption{Schematic diagram for determining a standard GPM sets that is U-equivalent to a given GPM set $\mathcal{M}$}
\label{fig5.1}
\end{figure}

Note that, for each $i_0\in \{1,\dots , n\}$, the unitary operator $R$ in Lemma \ref{lem5.1} is a local Clifford operator acting on the GPM set $M_{i_0}^{\dag }\mathcal{M}\triangleq \{ M_{i_0}^{\dag } M_i \}_{i=1}^n$. In order to find all standard GPM sets that are U-equivalent to a given GPM set $\mathcal{M}$, according to Lemma \ref{lem5.1}, it is sufficient to find all standard GPM sets that are UC-equivalent to one of the standard GPM sets $M_{i_0}^\dagger\mathcal{M}$. We denote all such GPM sets as $\mathcal{U}(\mathcal{M}$), which is the U-equivalence class of $\mathcal{M}$. As an extension of the Clifford-operator-based equivalence class in \cite{wang2025epj}, we provide the following procedure to find the U-equivalence class $\mathcal{U}(\mathcal{M})$.

\begin{procedure}[U-equivalent class $\mathcal{U}(\mathcal{M}$)]\label{proc5.1}
Let $\mathcal{M} =\{ (s_1,t_1), (s_2,t_2), \dots , (s_n,t_n) \}$ be an arbitrary $n$-GPM set.
\begin{enumerate}
\item[{\rm(1)}] Let $\mathcal{M}_{1}=\{(0,0),(s_2-s_1,t_2-t_1), \dots , (s_n-s_1,t_n-t_1) \}\approx(s_1,t_1)^{\dag}\mathcal{M}$. Then $\mathcal{M}_{1}$ is a standard GPM set.
\item[{\rm(2)}] Apply every local Clifford operator to $\mathcal{M}_{1}$ (such operators are determined by the classical representation (\ref{classical4.1})) and obtain all GPM sets that are UC-equivalent to $\mathcal{M}_{1}$. Denote the collection of all such standard  $n$-GPM sets as $\mathcal{UC}(\mathcal{M}_{1})$. Similarly, we can derive $\mathcal{UC}(\mathcal{M}_{i})$ for every $i\in \{1,\dots , n\}$.
\item[{\rm(3)}] The set $\mathcal{U(M)}\triangleq\cup_{i=1}^{n}\mathcal{UC}(\mathcal{M}_{i})$ is just the set of all standard $n$-GPM sets that are U-equivalent to the given set $\mathcal{M}$.
\end{enumerate}
\end{procedure}

The schematic diagram of Procedure \ref{proc5.1} is illustrated in Figure \ref{fig5.2}.

\begin{figure}[htbp]
\centering
\begin{tikzpicture}[
    node distance=0.8cm,
    op/.style={draw, rectangle, minimum width=1.2cm, minimum height=0.8cm, fill=blue!10},
    set/.style={ellipse, draw, minimum width=2cm, minimum height=1.2cm, fill=green!20},
    arrow/.style={-Stealth, thick}
]

% 输入集合 M
\node[set] (M) {$\mathcal{M} = \{M_1, M_2, \dots, M_n\}$};

% 中间集合 M_{i0}^†M
\node[set, below=of M] (MdaggerM) {$\mathcal{M}_i\approx M_{i}^{\dagger}\mathcal{M}$};
\draw[arrow] (M) -- node[right] {left-multiplied by $M_{i}^{\dagger}$} (MdaggerM);

% class
\node[set, below=of MdaggerM] (class) {$\mathcal{UC}(\mathcal{M}_{i})$};
\draw[arrow] (MdaggerM) -- node[right] {apply local Clifford operators to $\mathcal{M}_{i}$} (class);

% U-class
\node[set, below=of class] (Uclass) {$\mathcal{U(M)}\triangleq\cup_{i=1}^{n}\mathcal{UC}(\mathcal{M}_{i})$};
\draw[arrow] (class) -- node[right] {union} (Uclass);

\end{tikzpicture}
\caption{Schematic diagram for determining the U-equivalence class of a given GPM set $\mathcal{M}$}
\label{fig5.2}
\end{figure}

Based on Procedure \ref{proc5.1}, we propose a method to determine the U-equivalence of two $n$-GPM sets $\mathcal{M}$ and $\mathcal{N}$. Without loss of generality, we assume that $\mathcal{N}$ is a standard $n$-GPM set.

\begin{procedure}[U-equivalence]\label{proc5.2}
Let $\mathcal{M}= \{X^{s_1}Z^{t_1},\dots, X^{s_n}Z^{t_n} \}$ be an arbitrary $n$-GPM set, and $\mathcal{N}$ be a standard $n$-GPM set.
\begin{enumerate}
\item[{\rm(1)}] Apply Procedure \ref{proc5.1} to $\mathcal{M}$ and get all standard $n$-GPM sets that are U-equivalent to $\mathcal{M}$, that is, $\mathcal{U(M)}$.
\item[{\rm(2)}] Check whether the GPM set $\mathcal{N}$ is in $\mathcal{U(M)}$.
If yes, $\mathcal{M}$ and $\mathcal{N}$ are U-equivalent. 
If not, they are not U-equivalent.
\end{enumerate}
\end{procedure}

Using Matlab to implement  the Clifford-operator-based classification method, the authors of \cite{wang2025epj} obtained representative elements of 31 (Clifford-operator-based) equivalent  classes of 4-GBS sets  in $\mathbb{C}^6\otimes \mathbb{C}^6$. These elements are presented in Table \ref{tab5.1}, sorted in lexicographic order. Using Procedure \ref{proc5.2}, we show that the 31 representative 4-GBS sets are pairwise LU-inequivalent, thus forming 31 distinct LU-equivalence classes.

\begin{table}[htbp]\label{tab5.1}
\centering
\caption{Representative elements of 31 equivalent  classes of 4-GBS sets  in $\mathbb{C}^6\otimes \mathbb{C}^6$}
\label{tab5.1}
%\footnotesize
%\begin{ruledtabular}
\footnotesize
  \renewcommand{\arraystretch}{1.2}
\begin{tabular}{cccccc}
%\br
%\toprule
\toprule[1.5pt]  % 顶部粗线

%\midrule[0.8pt]  % 中间双线效果
\textbf{Set}&\textbf{Elements}&\textbf{Set}&\textbf{Elements}&\textbf{Set}&\textbf{Elements}\\
\hline
$\mathcal{S}_{1}$&$I,Z,Z^{2},Z^{3}$&$\mathcal{S}_{2}$&$I,Z,Z^{2},Z^{4}$&$\mathcal{S}_{3}$&$I,Z,Z^{2},X$\\
%\mr

$\mathcal{S}_{4}$&$I,Z,Z^{2},X^{2}$&$\mathcal{S}_{5}$&$I,Z,Z^{2},X^{2}Z$&$\mathcal{S}_{6}$&$I,Z,Z^{2},X^{3}$\\

$\mathcal{S}_{7}$&$I,Z,Z^{2},X^{3}Z$&$\mathcal{S}_{8}$&$I,Z,Z^{3},Z^{4}$&$\mathcal{S}_{9}$&$I,Z,Z^{3},X$\\

$\mathcal{S}_{10}$&$I,Z,Z^{3},X^{2}$&$\mathcal{S}_{11}$&$I,Z,Z^{3},X^{2}Z$&$\mathcal{S}_{12}$&$I,Z,Z^{3},X^{3}$\\

$\mathcal{S}_{13}$&$I,Z,Z^{3},X^{3}Z$&$\mathcal{S}_{14}$&$I,Z,Z^{3},X^{5}$&$\mathcal{S}_{15}$&$I,Z,X,XZ$\\

$\mathcal{S}_{16}$&$I,Z,X,XZ^{2}$&$\mathcal{S}_{17}$&$I,Z,X,XZ^{3}$&$\mathcal{S}_{18}$&$I,Z,X,X^{2}Z^{2}$\\

$\mathcal{S}_{19}$&$I,Z,X,X^{3}Z^{5}$&$\mathcal{S}_{20}$&$I,Z,X,X^{4}Z^{4}$&$\mathcal{S}_{21}$&$I,Z,X^{2},X^{2}Z$\\

$\mathcal{S}_{22}$&$I,Z,X^{2},X^{2}Z^{2}$&$\mathcal{S}_{23}$&$I,Z,X^{2},X^{2}Z^{5}$&$\mathcal{S}_{24}$&$I,Z,X^{2},X^{3}Z$\\

$\mathcal{S}_{25}$&$I,Z,X^{2},X^{4}$&$\mathcal{S}_{26}$&$I,Z,X^{3},X^{3}Z$&$\mathcal{S}_{27}$&$I,Z,X^{3},X^{3}Z^{3}$\\

$\mathcal{S}_{28}$&$I,Z,X^{3},X^{3}Z^{5}$&$\mathcal{S}_{29}$&$I,Z^{2},Z^{4},X^{2}$&$\mathcal{S}_{30}$&$I,Z^{2},X^{2},X^{2}Z^{2}$\\

$\mathcal{S}_{31}$&$I,Z^{3},X^{3},X^{3}Z^{3}$&$$&$$&$$&$$\\
\bottomrule[1.5pt]
\end{tabular}
%\end{ruledtabular}
\end{table}

\begin{example}\label{ex5.1}
The 31 GPM sets in Table \ref{tab5.1} is pairwise U-inequivalent. We derive this conclusion step by step using Procedure \ref{proc5.2}: (1) To determine the U-equivalence class of the first set $\mathcal{S}_{1}$ in Table \ref{tab5.1}, we apply Procedure \ref{proc5.1}\textemdash that is, we compute $\mathcal{U}(\mathcal{S}_{1})$. Note that $\mathcal{S}_{1}$ corresponds to the simple case described in Remark 1, where all local Clifford operators acting on it can be implemented through Clifford operators. Therefore, $\mathcal{U}(\mathcal{S}_{1})$ can be easily obtained using Procedure \ref{proc5.1} (see Table \ref{tab5.2}, we have omitted the common element $(0,0)$).  (2) By comparing Table \ref{tab5.2} with the other 30 sets in Table \ref{tab5.1}, we observe that none of them appear in Table \ref{tab5.2}. This conclusively demonstrates that $\mathcal{S}_{1}$ is U-inequivalent to all other 30 sets.
(3) For the remaining 30 sets in Table \ref{tab5.1}, by reapplying Procedure \ref{proc5.2} (Steps (1)--(2)), we conclusively demonstrate that all 31 GPM sets are mutually U-inequivalent.
\begin{table}[htbp]\label{tab5.2}
  \centering
  \caption{$\mathcal{U}(\mathcal{S}_{1})\triangleq\mathcal{UC}(\mathcal{S}_{1})\cup\mathcal{UC}((0,1)^{\dag}\mathcal{S}_{1})$ 
  (48 items)}
  \label{tab5.2}
  \footnotesize
  \renewcommand{\arraystretch}{1.25}
  \setlength{\tabcolsep}{4pt}
  \begin{tabular}{@{}>{\bfseries}l l@{}}
    \toprule[1.5pt]
    \multicolumn{1}{c}{\textbf{Elements of  class $\mathcal{UC}(\mathcal{S}_1)$ (24 items)}} \\
    %\midrule
    \begin{tabular}[t]{lll}
      $\{(0,1),(0,2),(0,3)\}$&$\{(0,2),(3,0),(3,4)\}$&$\{(0,2),(3,1),(3,3)\}$\\
      $\{(0,3),(0,4),(0,5)\}$&$\{(0,3),(2,0),(4,3)\}$&$\{(0,3),(2,2),(4,1)\}$\\
      $\{(0,3),(2,1),(4,2)\}$&$\{(0,3),(2,3),(4,0)\}$&$\{(0,3),(2,4),(4,5)\}$\\
      $\{(0,3),(2,5),(4,4)\}$&$\{(0,4),(3,0),(3,2)\}$&$\{(0,4),(3,3),(3,5)\}$\\
      $\{(1,0),(2,0),(3,0)\}$&$\{(1,1),(2,2),(3,3)\}$&$\{(1,2),(2,4),(3,0)\}$\\
      $\{(1,3),(2,0),(3,3)\}$&$\{(1,4),(2,2),(3,0)\}$&$\{(1,5),(2,4),(3,3)\}$\\
      $\{(3,0),(4,0),(5,0)\}$&$\{(3,0),(4,2),(5,4)\}$&$\{(3,0),(4,4),(5,2)\}$\\
      $\{(3,3),(4,0),(5,3)\}$&$\{(3,3),(4,2),(5,1)\}$&$\{(3,3),(4,4),(5,5)\}$
    \end{tabular} \\
  \hline
\multicolumn{1}{c}{\textbf{Elements of  class $\mathcal{UC}((0,1)^{\dag}\mathcal{S}_{1})$ (24 items)}} \\
    \midrule
    \begin{tabular}[t]{lll}
     $\{(0,1),(0,2),(0,5)\}$&$\{(0,1),(0,4),(0,5)\}$&$\{(0,2),(3,1),(3,5)\}$\\
     $\{(0,2),(3,2),(3,4)\}$&$\{(0,4),(3,1),(3,5)\}$&$\{(0,4),(3,2),(3,4)\}$\\
     $\{(1,0),(2,0),(5,0)\}$&$\{(1,0),(4,0),(5,0)\}$&$\{(1,1),(2,2),(5,5)\}$\\
     $\{(1,1),(4,4),(5,5)\}$&$\{(1,2),(2,4),(5,4)\}$&$\{(1,2),(4,2),(5,4)\}$\\
     $\{(1,3),(2,0),(5,3)\}$&$\{(1,5),(2,4),(5,1)\}$&$\{(1,5),(4,2),(5,1)\}$\\
     $\{(1,4),(4,4),(5,2)\}$&$\{(1,4),(2,2),(3,0)\}$&$\{(1,5),(2,4),(3,3)\}$\\
     $\{(2,0),(2,3),(4,3)\}$&$\{(2,1),(2,4),(4,5)\}$&$\{(2,1),(4,2),(4,5)\}$\\
     $\{(2,2),(2,5),(4,1)\}$&$\{(2,3),(4,0),(4,3)\}$&$\{(2,5),(4,1),(4,4)\}$
    \end{tabular}\\
   \hline
\multicolumn{1}{c}{\textbf{$\mathcal{UC}((0,2)^{\dag}\mathcal{S}_{1})=\mathcal{UC}((0,1)^{\dag}\mathcal{S}_{1})$}} \\
\hline
\multicolumn{1}{c}{\textbf{$\mathcal{UC}((0,3)^{\dag}\mathcal{S}_{1})=\mathcal{UC}(\mathcal{S}_{1})$}} \\
 \bottomrule[1.5pt]
\end{tabular}
\end{table}
\end{example}

Example \ref{ex5.1} shows that for 4-GPM sets on $\mathbb{C}^6$, U-equivalent classes coincide with Clifford-operator-based equivalent classes; that is, in this case, local Clifford operators provide no advantage over Clifford operators. Since there exist local Clifford operators that are not Clifford operators (such as the local Clifford operator $W$ in Lemma \ref{lem3.1}), the U-equivalent classes should be larger than the Clifford-operator-based equivalent classes. Thus, Example \ref{ex5.1} appears to be a special case. Through computation, we have found an example where the two classes share the same representative, yet the U-equivalent class is strictly larger than the Clifford-operator-based equivalent class.

\begin{example}
For the quantum system $\mathbb{C}^{3^4}$, let $\mathcal{M}=\{I, X^3, Z^{3} \}$. Then, by Procedure \ref{proc5.1},
\begin{eqnarray*}
\mathcal{M}_{1}&=& \mathcal{M},\\
\mathcal{M}_{2}&=&\{ I, X^{78}, X^{78}Z^{3} \}\approx (X^3)^\dag \mathcal{M},\\
\mathcal{M}_{3}&=&\{I, Z^{78}, X^3Z^{78} \}\approx (Z^{3})^\dag \mathcal{M}.
\end{eqnarray*}
According to Theorem \ref{th3.1}, the GPMs $X^{3}$ and $Z^{3}$ (in $\mathcal{M}$) are UC-equivalent to
$X^{78}Z^{3}$ and $X^{78}$ (in $\mathcal{M}_{2}$) respectively via the local Clifford operator with the classical representation
$\left[ \begin{array}{llll} 26 &26\\ 1 &0 \end{array} \right]$, where $26\times0-26\times1\equiv 1\pmod{3^{2}}$.
Similarly, the GPMs $X^{3}$ and $Z^{3}$ are UC-equivalent to
$Z^{78}$ and $X^{3}Z^{78}$ (in $\mathcal{M}_{3}$) respectively via the local Clifford operator with the classical representation
$\left[ \begin{array}{llll} 0 &1\\ 26 &26 \end{array} \right]$, where $0 \times 26-1\times 26\equiv 1\pmod{3^{2}}$.
Then the three UC-equivalent classes, $\mathcal{UC}(\mathcal{M})$, $\mathcal{UC}(\mathcal{M}_{2})$ and $\mathcal{UC}(\mathcal{M}_{3})$, are the same, and the U-equivalent class $\mathcal{U}(\mathcal{M})\triangleq\cup_{i=1}^{3}\mathcal{UC}(\mathcal{M}_{i})=\mathcal{UC}(\mathcal{M})$.
Let $\mathcal{N}=\{ I, X^{78}Z^{54}, X^{78}Z^{78} \}$, then the two GPMs $X^{3}$ and $Z^{3}$ are UC-equivalent to
$X^{78}Z^{54}$ and $X^{78}Z^{78}$ (in $\mathcal{N}$) respectively via the local Clifford operator with the classical representation 
$\left[ \begin{array}{llll} 26 &26\\ 18 &26 \end{array} \right]$, where $26\times 26-26\times 18=208\equiv 1\pmod{3^{2}}$ and $26\times 26-26\times 18=208\equiv 46\pmod{3^{4}}$.
So, by Theorem \ref{th3.1}, the local Clifford operator is not a Clifford operator. Therefore, the GPM set $\mathcal{N}$ belongs to  the U-equivalent class $\mathcal{U}(\mathcal{M})$ but does not belong to the Clifford-operator-based equivalent class $\mathcal{CU}(\mathcal{M})$. That is, the U-equivalent class $\mathcal{U}(\mathcal{M})$ is strictly contains the Clifford-operator-based equivalent class $\mathcal{CU}(\mathcal{M})$.
\end{example}

Using Matlab, it is known that the U-equivalent class $\mathcal{U}(\mathcal{M})$ contains $52,488$ standard GPM sets and the Clifford-operator-based equivalent class $\mathcal{CU}(\mathcal{M})$ contains $17,496$ standard GPM sets. That is, the U-equivalent class $\mathcal{U}(\mathcal{M})$ is much larger than the Clifford-operator-based equivalent class $\mathcal{CU}(\mathcal{M})$.

In this section, we present two procedures to show how to find the U-equivalence class of a given GPM set and then to determine whether any two given GPM sets are U-equivalent.  Generally, the U-equivalence class is larger than the Clifford-operator-based equivalence class. Two examples are provided to show that sometimes these two equivalence classes are identical, while in other cases, the U-equivalence class is strictly larger than the Clifford-operator-based equivalence class. The conditions under which these two equivalence classes coincide\textemdash that is, when local Clifford operators can be realized via Clifford operators\textemdash remain an open problem. This question is meaningful: if we know the precise conditions under which the equivalence classes coincide, we could fully achieve U-equivalence classification solely through the Clifford-operator-based classification. Such a breakthrough would advance solutions to quantum nonlocality problems, including the local discrimination of quantum states.

\section{Conclusion}
In this paper, we have extended the study of Clifford operators\textemdash which play a vital role in quantum information processing\textemdash to local Clifford operators, a class of unitary operators that map $n$-GPM sets to $n$-GPM sets. We have proven that local Clifford operators acting on nontrivial $2$-GPM sets admit a classical representation analogous to that of standard Clifford operators, providing a computationally convenient description via $2\times2$ matrices. Furthermore, we have demonstrated that any local Clifford operator acting on an $n$-GPM set can be decomposed into a product of standard Clifford operators and a local Clifford operator acting on a pair of GPMs. This decomposition offers a complete classical characterization of unitary conjugation mappings between $n$-GPM sets.

The classical representation of local Clifford operators is expected to find broad and significant applications in quantum information processing, much like the classical representation of standard Clifford operators. To this end, we have presented two procedures illustrating how local Clifford operators can be used to determine the U-equivalence class of a given GPM set\textemdash a more refined classification than the Clifford-operator-based equivalence classes previously described in \cite{wang2025epj}\textemdash and to assess whether two given GPM sets are U-equivalent.

A key open question remains: under what conditions do these two types of equivalence classes coincide? In other words, when can all local Clifford operators be realized via Clifford operators? Resolving this question would enable a complete U-equivalence classification solely through Clifford-operator-based methods, potentially advancing solutions to problems in quantum nonlocality, including the local discrimination of quantum states. It is expected that the classical representation of local Clifford operators introduced in this work would highlight studies on related challenges in quantum information processing.

\begin{acknowledgments}
This work is supported by NSFC (Grant No. 11971151, 12371132, 12075159, 12171044, 12171266, 12426650, 12426664); Wuxi University Research Start-up Fund for Introduced Talents; Fundamental Research Funds for the Central Universities and the specific research fund of the Innovation Platform for Academicians of Hainan Province.
\end{acknowledgments}

\section*{Appendix A: Proof of the necessity for the conditions in Lemma 2}

\begin{proof} Now we show that an arbitrary local Clifford operator preserves both the essential power of a GPM and the commutation relation between two GPMs.
For a binary GPM set $\{ X^{s_1}Z^{t_1}, X^{s_2}Z^{t_2} \}$, if there is a unitary matrix $U$ such that
\begin{eqnarray*}
UX^{s_1}Z^{t_1}U^\dag = e^{i\theta _1}X^{s'_1}Z^{t'_1},\ \ UX^{s_2}Z^{t_2}U^\dag = e^{i\theta _2}X^{s'_2}Z^{t'_2},
\end{eqnarray*}
then
\begin{align*}
U(X^{s_2}Z^{t_2}&X^{s_1}Z^{t_1}) U^\dag =(UX^{s_2}Z^{t_2}U^\dag)(UX^{s_1}Z^{t_1}U^\dag)\\
&=e^{i(\theta _1+\theta_2)} X^{s'_2}Z^{t'_2}X^{s'_1}Z^{t'_1}\\
&=e^{i(\theta _1+\theta_2)}\omega^{s'_1t'_2-t'_1s'_2}X^{s'_1}Z^{t'_1}X^{s'_2}Z^{t'_2}\\
&=\omega^{s'_1t'_2-t'_1s'_2}(UX^{s_1}Z^{t_1}U^\dag)(UX^{s_2}Z^{t_2}U^\dag)\\
&=\omega^{s'_1t'_2-t'_1s'_2}U(X^{s_1}Z^{t_1}X^{s_2}Z^{t_2})U^\dag.
\end{align*}
And $X^{s_2}Z^{t_2}X^{s_1}Z^{t_1}=\omega^{s_1t_2-t_1s_2} X^{s_1}Z^{t_1}X^{s_2}Z^{t_2}$.

Thus, we have
\begin{eqnarray}\label{indexrelation}
s'_1t'_2-t'_1s'_2\equiv s_1t_2-t_1s_2\ ({\rm mod} \ d).
\end{eqnarray}

Assume that the binary set we investigate has a form $M= \{ X^a, Z^b \}$, where $a|d, b|d$ and $0 < a, b < d$. If a unitary matrix
$U$ can realize the following transformation
\begin{eqnarray*}
UX^aU^\dag\approx X^sZ^t,\ \ UZ^bU^\dag \approx X^{s'}Z^{t'},
\end{eqnarray*}
by Lemma \ref{lem2.1}, $P_e(X^a)=P_e(X^sZ^t)$ and $P_e(Z^b)=P_e(X^{s'}Z^{t'})$, that is, $\gcd(s,t,d)=a$ and $\gcd(s',t',d)=b$. Hence the action of $U$ can
be written as follows.
\begin{eqnarray}\label{eq6}
UX^aU^\dag\approx X^{\mu a}Z^{\sigma a},\ \ UZ^bU^\dag \approx X^{\eta b}Z^{\nu b},
\end{eqnarray}
where
\begin{eqnarray}\label{eq7}
\gcd(\mu , \sigma , d/a)=1, \gcd(\eta , \nu ,d/b)=1,
\end{eqnarray}
\begin{eqnarray}\label{eq8}
(\mu\nu -\eta\sigma ) ab=ab \ ({\rm mod} \ d).
\end{eqnarray}
The last equality is due to Eq. \eqref{indexrelation}.
\end{proof}

\section*{Appendix B: Proof of Theorem \ref{th3.1}}

We need to prepare two lemmas to prove Theorem \ref{th3.1} in advance.
\begin{lemma}\label{lem2.1appb}
The $nm$-dimensional quantum system $\mathbb{C}^{nm}$ can be regarded as a bipartite system $\mathbb{C}^n \otimes \mathbb{C}^m$.
Then $X_{nm}^m=X_n\otimes I_m$, $Z_{nm}=Z_n\otimes \sum _i \omega _{mn}^i |i\rangle\langle i|$
and $Z_{nm}^n=I_n \otimes Z_m$, where $\omega _{mn}=e^{2\pi i/mn}$. The subscript denotes the dimensions of the corresponding quantum systems.
\end{lemma}
\begin{proof}
Let $|i+mj\rangle=|j\rangle\otimes |i\rangle (i\in\mathbb{Z}_{m}, j\in\mathbb{Z}_{n})$ be the computational basis of $\mathbb{C}^{nm}$.
Then
$X_{nm}^m|i+mj\rangle =|i+m(j+1)\rangle=(X_n\otimes I_m)|i+mj\rangle$, and
$Z_{nm}|i+mj\rangle =\omega _{mn}^{i+mj}|i+mj\rangle=Z_n|j\rangle\otimes \omega _{mn}^i|i\rangle
=(Z_n\otimes \sum _i \omega _{mn}^i |i\rangle\langle i|)|i+mj\rangle.$
\end{proof}

\begin{lemma}\label{lem2appb}
Let $a$ and $b$ are both positive factors of $d$.
If $a|b$, $\gcd(u,v,d/a)=1, uab \equiv ab \ ({\rm mod} \  d)$ and $\gcd(u, d/a, d/b)=1$,
then $X^{ua}Z^{va}$ and $Z^{b}$ are UC equivalent to $X^{a}$ and $Z^{b}$ respectively,
that is, $\{X^a, Z^b \} \sim \{X^{ua}Z^{va}, Z^b \}$.
\end{lemma}
\begin{proof}
It is easy to see that $X^{a}$ and $Z^{b}$ are UC equivalent to $Z^{-a}$ and $X^{b}$ respectively
via the Clifford operator with its symplectic representation
$\left[ \begin{array}{ccc} 0 &1\\ -1 &0 \end{array} \right]$, that is,
\begin{eqnarray}\label{eq1lem2ap2}
\{X^a, Z^b \}\mathop{\sim} \{Z^{-a},  X^b \}.
\end{eqnarray}
Since $\gcd(d/a-u,d/a-v,d/a)=\gcd(u,v,d/a)=1$, by Lemma \ref{lem2.1},
the GPMs $X^{(d/a-u)a}Z^{(d/a-v)a}$ and $Z^b$ are UC equivalent to $Z^{a}$ and $X^{u^{\prime}b}Z^{v^{\prime}b}$ respectively
via the Clifford operator $C_{\gcd(d/a-u, d/a-v, d/a)}$,
where $u^{\prime}\equiv d/a-u$ (mod $d/b$).
That is,
\begin{eqnarray}\label{eq2lem2ap2}
\{X^{(d/a-u)a}Z^{(d/a-v)a}, Z^b \}\mathop{\sim}\{Z^a,  X^{u^{\prime}b}Z^{v^{\prime}b} \}.
\end{eqnarray}
Then $\gcd(u^{\prime}, d/a, d/b)=\gcd(d/a-u, d/a, d/b)=\gcd(u, d/a, d/b)=1$.
Furthermore, $\gcd(u^{\prime}, d/b)=\gcd(u^{\prime}, d/a, d/b)=1$ follows by the condition $a|b$.
According to  Lemma \ref{lem3.1}, the GPMs $X^{b}$ and $Z^a$ are UC equivalent to $X^{u^{\prime}b}$ and $Z^{a}$ respectively.
By the equation $\gcd(u^{\prime}, d/b)=1$, there exists an integer $c$ such that $cu^{\prime}\equiv 1$ (mod $d/b$).
Then the two GPMs $X^{u^{\prime}b}$ and $Z^{a}$ are UC equivalent to $X^{u^{\prime}b}Z^{v^{\prime}b}$ and $Z^{a}$ respectively
through a Clifford operator with symplectic representation
$V = \left[ \begin{array}{llll} 1 &0\\ cv^{\prime} &1 \end{array} \right].$
So the GPMs $X^{b}$ and $Z^a$ are UC equivalent to $X^{u^{\prime}b}Z^{v^{\prime}b}$ and $Z^{a}$ respectively,
that is,
\begin{eqnarray}\label{eq3lem2ap2}
\{Z^a, X^{b} \}\mathop{\sim}\{Z^a,  X^{u^{\prime}b}Z^{v^{\prime}b} \}.
\end{eqnarray}
From the transitivity of UC-equivalence and Eqs. \eqref{eq1lem2ap2}-\eqref{eq3lem2ap2}, we obtain the conclusion.
\end{proof}

Now we can show the proof of Theorem \ref{th3.1}.

\begin{proof}
First, we prove the necessity of the conditions in Theorem \ref{th3.1}. Suppose that two GPMs $X^{a}$ and $Z^{b}$ are UC equivalent to $X^{u_{1} a}Z^{v_{1} a}$ and $X^{u_{2} b}Z^{v_{2} b}$ respectively. Since UC-equivalent GPMs share the same essential powers and UC transformations preserve the commutation coefficient,
Eqs. \eqref{eq1th3.1} and \eqref{eq2th3.1} hold.
We now proceed to prove Eq. \eqref{eq3th3.1}.
Lemma \ref{lem2.1} and Eq. \eqref{eq1th3.1} ensure that the two GPMs $X^{u_{1} a}Z^{v_{1} a}$ and $X^{u_{2} b}Z^{v_{2} b}$
are UC equivalent to $X^{ua}Z^{va}$ and $Z^{b}$ respectively through a Clifford operator $C_{\gcd(u_{2}, v_{2}, d/b)}$. A straightforward calculation shows that
$u\equiv u_{1}v_{2}-u_{2}v_{1}$ (mod $d/b$).
Then $\gcd(u, v, d/a)=\gcd(u_{1},v_{1},d/a)=1$, $uab\equiv(u_{1}v_{2}-u_{2}v_{1})ab\equiv ab\ ({\rm mod}\ d)$
and $\gcd(u, d/a, d/b)=\gcd(u_{1}v_{2}-u_{2}v_{1}, d/a, d/b)$.
According to the transitivity of UC equivalence,
the GPMs $X^{ua}Z^{va}$ and $Z^{b}$ are UC equivalent to $X^{a}$ and $Z^{b}$ respectively
through a unitary operator $H$, that is,
\begin{eqnarray}\label{eq1proofth1}
X^{a}\mathop{\sim}\limits^{H}X^{ua}Z^{va},\ \ Z^{b} \mathop{\sim}\limits^{H} Z^{b}.
\end{eqnarray}
Let $c=\gcd(u, d/a)$ and $d=ace$, Eq. \eqref{eq1proofth1} leads to the following equation
\begin{eqnarray}\label{eq2proofth1}
HX^{ae}H^\dag=f(\omega) X^{uae}Z^{vae}=f(\omega)Z^{vae},
\end{eqnarray}
where $f(\omega)$ is some complex number with modulus 1.
Recall that the commutation coefficient of two GPMs, $X^{ae}$ and $Z^{b}$,  are invariant through UC-transformations,
we have $aeb\equiv 0\ ({\rm mod} \  d)$, and then $c$ is a factor of $b$, that is, $c|b$.
Let $b=b_1c$, then $b_1|(ae)$ follows by $b|d$ and $d=ace$.
Now we have
\begin{equation*}\label{eqaeb}
[ae,b]=\frac{aeb}{\gcd(ae,b)}=\frac{aec}{\gcd(ae/b_1,c)}=\frac{d}{\gcd(ae/b_1,c)},
\end{equation*}
where $[ae,b]$ refers to the least common multiple of $ae$ and $b$, so $[ae, b]$ is also a factor of $d$.
It can be proven that $[ae, b]=d$.
If it is not the case, we can assume that $[ae, b]=aed_1=bd_2 (<d)$, then, by Eqs. \eqref{eq1proofth1} and \eqref{eq2proofth1},
\begin{eqnarray*}
X^{[ae, b]}&=X^{aed_1}\mathop{\sim}\limits^{H} Z^{vaed_1}=Z^{v[ae, b]},\\
Z^{v[ae, b]}&=Z^{vbd_2}\mathop{\sim}\limits^{H} Z^{vbd_2}=Z^{v[ae, b]}.
\end{eqnarray*}
Thus, $X^{[ae, b]}\approx Z^{v[ae, b]}$. This is impossible.
Therefore $[ae, b]=d$ and $\gcd(ae/b_1,c)=1$, and  then $\gcd(u, d/a, d/b)=\gcd(c, ae/b_1)=1$ follows by $b=b_1c$ and $d=ace$.
Now we have $\gcd(u_{1}v_{2}-u_{2}v_{1}, d/a, d/b)=\gcd(u, d/a, d/b)=1$ and Eq. \eqref{eq3th3.1} follows.

Next, we prove the sufficiency of the conditions. Assume Eqs.\eqref{eq1th3.1}-\eqref{eq3th3.1} are true, Eq.\eqref{eq1th3.1} ensures that the two GPMs $X^{u_{2} b}Z^{v_{2} b}$ and $X^{u_{1} a}Z^{v_{1} a}$ are UC equivalent to $Z^{b}$ and $X^{ua}Z^{va}$ respectively through a Clifford operator $C_{\gcd(u_{2}b,v_{2}b,d)}$,
where $u\equiv u_{1}v_{2}-u_{2}v_{1}$ (mod $d/b$). Then $\gcd(u, v, d/a)=\gcd(u_{1},v_{1},d/a)=1$, $uab\equiv(u_{1}v_{2}-u_{2}v_{1})ab\equiv ab\ (mod\ d)$
and $\gcd(u, d/a, d/b)=\gcd(u_{1}v_{2}-u_{2}v_{1}, d/a, d/b)$. Therefore, in order to prove the sufficiency, we only need to prove that if $\gcd(u,v,d/a)=1$, $uab \equiv ab (mod \  d)$ and $\gcd(u, d/a, d/b)=1$, then $X^{a}$ and $Z^{b}$ are UC equivalent to $X^{ua}Z^{va}$ and $Z^{b}$ respectively.
Let $\gcd(a, d/b)=a_1$, then $\gcd(a/a_1, d/(ba_1))=1$, $d=a_1be$ and $a=a_1m$ where $e$ and $m$ are two positive integers.
And then $\gcd(m, e)=1$, $m|(be)$ and $m|b$. Note that $C^{d}= C^{be} \otimes C^{a_1}$, Lemma \ref{lem2.1appb} yields
$X^a=(X^{a_1})^m=X_{be}^m \otimes I_{a_1}, Z^b=Z_{be}^b \otimes \sum _i \omega _{d}^{bi} |i\rangle\langle i|,$
\begin{eqnarray*}
X^{ua}Z^{va}=X_{be}^{um}Z_{be}^{va}\otimes \sum _i \omega _{d}^{vai} |i\rangle\langle i|\\
=\sum _i \omega _{be}^{vmi}X_{be}^{um}Z_{be}^{va}\otimes |i\rangle\langle i|.
\end{eqnarray*}
Meanwhile, $umb\equiv mb\ (mod\ be)$ follows by $uab\equiv ab\ (mod\ d)$ and $a=a_1m$, and $\gcd(u, d/a, d/b)=1$ leads to $\gcd(u, d/a, a_1)=1$ and $\gcd(u, d/a, e)=1$.
So that $\gcd(u,va_1,be/m)=\gcd(u,v,d/a)\gcd(u, d/a, a_1)=1$ and $\gcd(u, be/m, be/b)=\gcd(u, d/a, e)=1$. According to $m|b$ and Lemma \ref{lem2appb}, we have
$$\{X_{be}^{m}, Z_{be}^b\}\mathop{\sim}\{X_{be}^{um}Z_{be}^{va_1m}, Z_{be}^b\}=\{X_{be}^{um}Z_{be}^{va}, Z_{be}^b\}.$$
Finally, we derive the conclusion
\begin{eqnarray*}
\{X^{a}, Z^b\}=\{X_{be}^m \otimes I_{a_1}, Z_{be}^b \otimes \sum _i \omega _{d}^{bi} |i\rangle\langle i|\}\\
\mathop{\sim}\limits^{\sum _i Z_{be}^{vi}\otimes |i\rangle\langle i|}
\{\sum _i \omega _{be}^{vmi}X_{be}^{m}\otimes |i\rangle\langle i|,Z_{be}^b \otimes \sum _i \omega _{d}^{bi} |i\rangle\langle i|\}\\
\sim \{\sum _i \omega _{be}^{vmi}X_{be}^{um}Z_{be}^{va}\otimes |i\rangle\langle i|,Z_{be}^b \otimes \sum _i \omega _{d}^{bi} |i\rangle\langle i|\}\\
=\{X^{ua}Z^{va}, Z^b\}.
\end{eqnarray*}
\end{proof}

\section*{Appendix C: Lemma \ref{lemappd}}

\begin{lemma}\label{lemappd}
Let $\mathcal{M}'= \{Z^b, X^{s_2}Z^{t_2}\cdots , X^{s_n}Z^{t_n} \}$ be a $n$-GPM set satisfying $b|d$ and $b|t_i (i=2, \dots , n)$, then $U$ is a local Clifford operator acting on $\mathcal{M}'$ if and only if $U$ is a local Clifford operator acting on $\{X^a, Z^b \}$ ($a=\gcd(s_2, \cdots , s_n, d)$).
\end{lemma}

\begin{proof}
Since $b|t_i$, it is obvious that $U$ is a local Clifford operator acting on $\mathcal{M}'$ if and only if $U$ is a local Clifford operator acting on the $n$-GPM set $\{Z^b, X^{s_2}, \dots, X^{s_n}\}$. If $U$ is a local Clifford operator on two GPMs $X^a$ and $Z^b$, then it is clear that $U$ is a local Clifford operator on the $n$-GPM set $\{Z^b, X^{s_2}, \dots, X^{s_n}\}$.
Conversely, if $U$ is a local Clifford operator on the $n$-GPM set $\{Z^b, X^{s_2}, \dots, X^{s_n}\}$, noting that $a=\gcd(s_2, \cdots , s_n, d)$, there exist integers $k_2, \cdots , k_n$ such that $k_2s_2+ \cdots + k_ns_n\equiv a$ (mod $d$), then $UX^aU^\dag=UX^{k_2s_2+ \cdots + k_ns_n}U^\dag=(UX^{s_2}U^\dag)^{k_2}\cdots(UX^{s_n}U^\dag)^{k_n}$ is a GPM. So $U$ is a local Clifford operator on $X^a$ and $Z^b$.
\end{proof}

\section*{Appendix D: Proof of Equations (\ref{classical4.5.1})-(\ref{classical4.5.2})}
\begin{proof} Let $P_e(\mathcal{M})=(p, p, \cdots p, q, \dots , q)$ where the first $m$ terms are the prime number $p$ and the remaining $n-m$ terms are the prime number $q$. Applying the Clifford operator $C_{\gcd(s_1,t_1,pq)}$ to $\mathcal{M}$, we obtain
$$\mathcal{M}'=\{Z^p, X^{s'_2}Z^{t'_2}, \cdots , X^{s'_n}Z^{t'_n} \} .$$
If $t'_{m+1}=\cdots =t'_n=0$, then $\mathcal{M}'$ satisfies Condition (I) and the local Clifford operators on $\mathcal{M}$ have the form $L_{(a,p)}\circ C_{\gcd(s_1,t_1,pq)}$.
If at least one of $t'_{m+1},\cdots ,t'_n$ is non-zero, by applying the division algorithm, we get 
$$\mathcal{M}_1=\{Z^p, X^{s'_2}, \dots ,X^{s'_m}, X^{s'_{m+1}}Z^{r_{m+1}}, \dots , X^{s'_n}Z^{r_n} \},$$
where $0\leq r_{i}<p, m+1\leq i\leq n $.
Without loss of generality, assume $r_n\neq 0$, then $\gcd(s'_n, r_n, p)=\gcd(r_n, q)=1$ and $Pe(X^{s'_n}Z^{r_n})=1$. Hence, the local Clifford operators on $\mathcal{M}$ have the form
$$L_{(a,1)}\circ C_{\gcd(s'_n, r_n, pq)}\circ C_{\gcd(s_1,t_1,pq)}.$$
\end{proof}

\nocite{*}

\end{document}